\documentclass[abstracton,a4paper]{scrartcl}
\usepackage[utf8]{inputenc}
\usepackage[T1]{fontenc}
\usepackage{graphicx}
\usepackage[usenames,dvipsnames,svgnames]{xcolor}
\usepackage{amsmath}
\usepackage{amssymb}
\usepackage{amsfonts}
\usepackage{amsthm}
\usepackage{xspace}
\usepackage{enumerate}
\usepackage[font=small]{caption}
\usepackage{subcaption}
\usepackage{wrapfig}

\usepackage[olditem,oldenum]{paralist} 

\let\oldthebibliography\thebibliography
\renewcommand{\thebibliography}[1]{\oldthebibliography{#1}\addtolength{\itemsep}{-0.5pt}}

\usepackage{hyperref}

\hypersetup{%
    colorlinks=true, linktocpage=true, pdfstartpage=1, pdfstartview=FitV,%
    breaklinks=true, pdfpagemode=UseNone, pageanchor=true, pdfpagemode=UseOutlines,%
    plainpages=false, bookmarksnumbered, bookmarksopen=true, bookmarksopenlevel=1,%
    hypertexnames=true, pdfhighlight=/O,
    urlcolor=darkred, linkcolor=lightblue, citecolor=darkgreen, 
    pdftitle={},%
    pdfauthor={},%
    pdfsubject={},%
    pdfkeywords={},%
    pdfcreator={pdfLaTeX},%
    pdfproducer={LaTeX}%
}

\addtokomafont{disposition}{\rmfamily}

\input{maketitle}

\definecolor{darkred}{rgb}{0.5,0,0}
\definecolor{lightblue}{rgb}{0,0.4,0.8}
\definecolor{darkgreen}{rgb}{0,0.5,0}

\usepackage{aliascnt}
\usepackage{xparse}
\NewDocumentCommand{\xnewtheorem}{m o m}{\IfNoValueTF{#2}{\newtheorem{#1}{#3}\expandafter\def\csname #1autorefname\endcsname{#3}}{
\newaliascnt{#1}{#2}
\newtheorem{#1}[#1]{#3}
\aliascntresetthe{#1}
\expandafter\def\csname #1autorefname\endcsname{#3}
}
}

\newtheorem{theorem}{Theorem}[section]
\xnewtheorem{lemma}[theorem]{Lemma}
\xnewtheorem{proposition}[theorem]{Proposition}
\xnewtheorem{corollary}[theorem]{Corollary}
\xnewtheorem{observation}{Observation}
\xnewtheorem{fact}{Fact}
\xnewtheorem{example}{Example}
\xnewtheorem{definition}{Definition}
\xnewtheorem{problem}{Problem}

\xnewtheorem{boldclaim}{Bold claim}
\xnewtheorem{claim}{Claim}

\theoremstyle{remark}

\newcommand{\ceil}[1]{\ensuremath{\lceil#1\rceil}}
\newcommand{\floor}[1]{\ensuremath{\lfloor#1\rfloor}}

\newcommand{\ctime}{\ensuremath{\mathfrak{T}}\xspace}

\DeclareMathOperator{\BIGO}{O}

\DeclareMathOperator{\BIGOMEGA}{\Omega}

\newcommand{\BigO}[1]{\ensuremath{\BIGO\left(#1\right)}}

\newcommand{\BigOmega}[1]{\ensuremath{\BIGOMEGA\left(#1\right)}}

\newcommand{\rd}{deterministic binary majority process\xspace}
\newcommand\qswap{$q$-swap\xspace}
\newcommand\fund[1]{\ifx&#1& \else $#1$-\fi permanent\xspace}
\newcommand\qfund{\fund{q}}

\newcommand\dtime{voting time\xspace}
\newcommand\Dtime{Voting Time\xspace}
\newcommand\contime{convergence time\xspace}
\newcommand{\conv}{converge\xspace}
\newcommand{\stab}{two-periodic\xspace}
\newcommand{\sstate}{two-periodic state\xspace}
\newcommand{\stabil}{converge\xspace}

\newcommand\Veven{\ensuremath{V_\text{{\normalfont even}}}\xspace}
\newcommand\Vodd{\ensuremath{V_\text{{\normalfont odd}}}\xspace}

\newcommand\ifrac[2]{#1/#2}

\newcommand\ctdp{\textsc{vtdp}\xspace}

\newcommand{\family}[1]{\ensuremath{\operatorname{fam}\left(#1\right)}}

\title{On the \Dtime of the\\ Deterministic Majority Process}

\author{ Dominik Kaaser \\ University of Salzburg \\ \emph{dominik@cosy.sbg.ac.at} \and
    Frederik Mallmann-Trenn \\ École Normale Supérieure\\ Simon Fraser University \\ \emph{fmallman@sfu.ca} \and
    Emanuele Natale
\\ Sapienza Università di Roma \\ \emph{natale@di.uniroma1.it}}
\date{}

\begin{document}

\nocite{KMN15}

\maketitle
\begin{abstract}
In the \rd we are given a simple graph where each node has one out of two initial
opinions. In every round, every node adopts the majority opinion among its
neighbors. By using a potential argument first discovered by Goles and Olivos (1980),
it is known that this process always \conv{s} in $\BigO{|E|}$ rounds to a
two-periodic state in which every node either keeps its opinion or changes
it in every round.

It has been shown by Frischknecht, Keller, and Wattenhofer (2013) that
the $\BigO{|E|}$ bound on the \contime of the \rd is indeed tight even for
dense graphs. However, in many graphs such as the complete graph, from any
initial opinion assignment, the process \conv{s} in just a constant number of
rounds.

By carefully exploiting the structure of the potential function by Goles and
Olivos (1980), we derive a new upper bound on the \dtime of the \rd
that accounts for such exceptional cases. We show that it is possible to
identify certain modules of a graph $G$ in order to obtain a new graph
$G^\Delta$ with the property that the worst-case \contime of $G^\Delta$
is an upper bound on that of $G$. Moreover, even though our upper bound can be
computed in linear time, we show that, given an integer $k$, it is NP-hard to decide
whether there exists an initial opinion assignment for which it takes more
than $k$ rounds to \conv{} to the two-periodic state.

\end{abstract}

\section{Introduction}

We study the \emph{\rd} which is defined as follows. We are given a graph $G = (V, E)$
where each node has one out of two opinions. The process runs in discrete
rounds where each node in every round computes and adopts the majority opinion among all of
its neighbors.

It is known that this deterministic process always \conv{s} to a two-periodic state. 
The \emph{\contime} of a given graph for a given
initial opinion assignment is the time it takes until the
\sstate{} is reached. 
In this work we give bounds on the \emph{\dtime},
which is the maximum \contime over all possible initial opinion assignments.

The \rd  has widespread
applications in the study of \emph{influence networks} in distributed computing \cite{FKW13}, 
distributed databases \cite{Gif79}, sensor networks \cite{BTV09}, the
competition of opinions in social networks \cite{MT14}, social behavior in game
theory \cite{DP94}, chemical reaction networks \cite{AAE07}, neural and
automata networks \cite{G90}, and cells' behavior in biology \cite{CC12}. 
Variants of the \rd have been used
in the area of distributed
community detection \cite{RAK07,KPS13,CG10}, where the \dtime is essentially 
the convergence time of the proposed community-detection protocols. 

Among its many probabilistic variants that have been previously considered, plenty of work concerns \emph{randomized voting} where in each step every node 
is allowed to contact a random sample of its neighbors and updates its current 
opinion according to the majority opinion in that sample
\cite{AF14,BMPS05,CEOR12,DW83,HL75,HP01,L85,LN07,Mal14,O12}.

In an algorithmic game theoretic setting, the \rd can be seen as the simplest
\emph{discrete preference games} \cite{CKO13}. In this game theoretic
perspective, the existence of so-called monopolies has been investigated
\cite{ACF14}. A \emph{monopoly} in a graph is a set of nodes which 
start with the same opinion, causing all other nodes to eventually adopt this opinion. 
In the distributed computing area, a lot of research has been done to find small
monopolies, see for example \cite{P02}. It has also been shown that there exist families of
graphs with constant-size monopolies \cite{Ber01}. More recently,
classes of graphs which do not have small monopolies have been investigated
\cite{P14}.

Many of these results relate to the \dtime of the \rd. It was proven
independently by Goles and Olivos
\cite{GO80}, and Poljak and Sůra \cite{PS83} with the same potential function argument that the \rd always \stabil{s} to a \sstate.
They later (independently) refined and
generalized the potential function argument in several directions
\cite{GFP85,PT86,GO88,G89}. Their proof was popularized in the \emph{Puzzled}
columns of Communications of the ACM \cite{Win08b,Win08}. Recently, the same
problem has been studied on infinite graphs w.r.t.\ a given probability
distribution on the initial opinion assignments \cite{BCO14}. In \cite{TT13}, the authors provide a bound on the number of times a node in a given bounded-degree graph changes its opinion. Both \cite{BCO14} and \cite{TT13} also investigate the probability that in the \sstate all nodes hold the same opinion. 

As for the maximum time it takes for the process to \stabil over all initial opinion assignments, Frischknecht et al.\ \cite{FKW13} note that the
potential argument by Goles et al.\ \cite{GO80, PS83, Win08} can be used to prove an
$\BigO{|E|}$ upper bound. They furthermore show
that this upper bound is tight in general, by designing a class of
graphs in which the \rd takes at least $\BigOmega{|V|^2}$ rounds to \stabil from a given initial opinion assignment. This construction
has later been extended to prove lower bounds for weighted and multi-edges
graphs by Keller et al.\ \cite{KPW14}.

A lot of attention has been given to the \sstate to which the \rd \conv{s} to. However, 
besides the $\BigO{|E|}$ upper bound that follows from the result by Goles et al.\ \cite{GO80, PS83, Win08},
no further upper bound on the \dtime that holds for any initial opinion assignment has been proved.
Still, one can observe that in many graphs the \dtime is much smaller than $\BigO{|E|}$, e.g., the \dtime of the complete graph is one.

We show that for the \rd the
question whether the \dtime is greater than a given number is NP-hard. While
for many generalizations of the \rd many decision problems are known to be
NP-hard, at the best of our knowledge this is the first NP-hardness proof that
does not require any additional mechanisms besides the bare majority rule of
the \rd.
However, as we show in the rest of the paper, it is possible to obtain upper bounds on the \dtime which can be computed in linear time.
A module of a graph is a subset of vertices $S$ such that for
each pair of nodes $u,v\in S$ it holds that $N(u)\setminus S=N(v)\setminus S$.
By carefully exploiting the structure of the potential function by Goles et
al.\ and leveraging the particular behavior that certain modules of the
graph exhibit in the \rd,  we are able to prove that the \dtime of a graph can be
bounded by that of a smaller graph that can be constructed in linear time by
contracting suitable vertices. 

We obtain a new upper bound that asymptotically improves on the  previous
$\BigO{|E|}$ bound on graph classes which are characterized by a high number of
modules that are either cliques or independent sets.  For instance, the Turán
graph $T(n,r)$ is the graph formed by partitioning a set of $n$ vertices into
$r$ subsets, with sizes as equal as possible, and connecting two vertices by an
edge whenever they belong to different subsets.  For the \contime of the Turán
graph $T(n,r)$ we obtain an $\BigO{r^2}$ bound, compared to the previously best
known bound of $\BigO{n^2}$. Also, for the \contime of full $d$-ary trees we
get an  $\BigO{|V|/d}$ bound, compared to $\BigO{|V|}$ originating from the
$\BigO{|E|}$ bounds. Our bound relies on a well-known graph contraction
technique, e.g., see the notion of \emph{identical vertices} in \cite{SSKC14}.

\subsection{Preliminaries}

We are given a connected graph $G = (V, E)$ and an initial opinion assignment which we define now.
\begin{definition}
An opinion assignment $f_t$ in round $t \geq 0$ is a function $f_t : V \rightarrow
\{0, 1\}$ which assigns for each $v \in V$ an opinion  with
\begin{equation*}
f_t(v) =
	\begin{cases}
		1  & \text{if $v$ has opinion $1$ at time $t$}\\
		0  & \text{if $v$ has opinion $0$ at time $t$.}
	\end{cases}
\end{equation*}
We will also denote opinion $1$ as \emph{white} and opinion $0$ as \emph{black}. 
The opinion assignment at time $t = 0$ is called \emph{initial opinion assignment}.
\end{definition}

The \rd can be defined as follows. Let $v$ be an arbitrary
but fixed vertex and $N(v)$ the set of neighbors of $v$. To compute
$f_{t+1}(v)$ the node $v$ computes the majority opinion of all of its neighbors
in $N(v)$. In the case of a tie the node behaves \emph{lazily}, that is, $v$ stays
with its own opinion. Otherwise, there is a \emph{clear majority} and the
node adopts the majority opinion. This leads to the following
definition.

\begin{definition} \label{def:majority-voting}
Let $G = (V, E)$ be a graph and let $f_0$ be an initial opinion assignment such that $f_0: V
\rightarrow \{0, 1\}$. The \rd is the
series of opinion assignments that satisfies the following
rule.
\begin{equation*}\
f_{t+1}(v) = \begin{cases}
		0	& \text{\normalfont if } \left|\left\{u \in N(v): f_t(u) = 0 \right\}\right| > \left|\left\{u \in N(v): f_t(u) = 1 \right\}\right| \\
		1	& \text{\normalfont if } \left|\left\{u \in N(v): f_t(u) = 0 \right\}\right| < \left|\left\{u \in N(v): f_t(u) = 1 \right\}\right| \\
	f_t(v)	& \text{\normalfont otherwise}\\
\end{cases}
\end{equation*}
\end{definition}

Note that the pair $\left(G,f_0 \right)$ completely
determines the behavior of the system according to the majority process. 
We next define the main object of this work, namely the voting time.

\begin{definition}
Given a graph $G=(V,E)$ and any initial opinion assignment $f_0$ on $V$, the
\emph{\contime} \ctime of the majority process on $G$ w.r.t.\ $f_0$ is
$\ctime = \ctime(G,f_0) = \min \left\{t:\forall v ~ f_{t+2}(v)=f_t(v) \right\}.$
The \emph{\dtime} of $G$ is defined as $\displaystyle \max_{ f_0 \in \{ 0,1\}^{V} } \ctime(G,f_0)$.
\end{definition}

Observe that $\ctime$ is indeed the number of steps until the
process \conv{s} to a periodic state.
This holds since the process is completely determined by the current opinion assignment. Thus $f_{t+2}(v)=f_t(v)$ also implies that $f_{t+3}(v)=f_{t+1}(v)$ for all nodes $v$. 

In the following we assume without loss of generality that $G$ is connected.
For disconnected graphs the \rd runs independently in each
connected component. Therefore, the resulting upper bounds on the
\dtime time can be replaced by the maximum over the corresponding bounds
in the individual connected components of $G$.

\subsection{Our Contribution}
First we define the \emph{\dtime decision problem} \ctdp and
show that it is NP-complete. 

\begin{definition}[\dtime decision problem, \ctdp]
For a given graph $G$ and an integer $k$, is there an assignment of initial
opinions such that the \dtime of $G$ is at least $k$?
\end{definition}

\begin{theorem}\label{thm:hardness}
Given a general simple graph $G$, \ctdp is NP-complete.
\end{theorem}

Then, in \autoref{sect:main-result} we extend known approaches 
to derive upper bounds on the \dtime, which are tight for general graphs. 
In \autoref{sect:symmetry}, we identify the following subsets of nodes that play a crucial role in determining the \dtime of the \rd. 
\begin{definition}
A set of nodes $S$ is called a \textit{family} if and only if for all pairs of nodes $u,v\in S$ we have
$N\left( u \right)\setminus\left\{ v \right\}=N\left( v \right)\setminus\left\{ u \right\}$. We say that a family $S$ is \emph{proper} if $|S|>1$.
\end{definition}
The set of families of a graph forms a partition of the nodes into equivalence classes. 
Our main~contribution is a proof that the \dtime of the \rd is bounded by that
of a new graph obtained by contracting its families into one or two nodes, as stated in the following
theorem. 

\begin{definition}
Given a graph $G=(V,E)$, its \textit{asymmetric graph} 
$G^{\Delta}=(V^{\Delta},E^{\Delta})$ is the subgraph of $G$ induced
by the subset $V^\Delta \subseteq V$ constructed by replacing every family of odd-degree 
non-adjacent nodes with one node and replacing any other proper family with two nodes.
\end{definition}

\begin{theorem} \label{thm:asym_upper}
Given any initial opinion assignment on a graph $G=(V,E)$, the \dtime
of the \rd is at most \[ 1+\min \left \{
|E^{\Delta}|-\frac{|V^{\Delta}_{odd}|}2,
\frac{|E^{\Delta}|}2+\frac{|V^{\Delta}_{even}|}4+\frac{7}{4}\cdot|V^\Delta|\right \} \enspace . \]
Furthermore, this bound can be computed in $\BigO{|E|}$ time.
\end{theorem}
As mentioned before, this bound becomes $\BigO{r^2}$ for the Turán graph $T(n,r)$ and $\BigO{|V|/d}$ for $d$-ary trees.
Finally, in \autoref{sec:furthercompprops}  of the appendix, we give some insight into the computational properties of the \dtime.


\newcommand\megaclique{mega-clique\xspace}
\newcommand\MegaClique{Mega-Clique\xspace}

\section{NP-Completeness}\label{sec:np}

If it was possible to efficiently compute the worst-case \dtime, 
there would have been not much interest in investigating good upper bounds for it.
In this section, we show that this is unlikely to be the case. 
We prove \autoref{thm:hardness} by reducing \textsc{3sat} to the \dtime decision
problem.  
Given $\Phi\in$ \textsc{3sat}, we construct a graph $G =
G(\Phi)$ such that the \rd on $G$ simulates the evaluation of $\Phi$. The graph
$G$ consists of $h$ layers. The first layer represents an assignment of the
variables in $\Phi$, the remaining layers represent $\Phi$ and ensure that the
assignment of variables in $\Phi$ is valid. We will show that if $\Phi$ is
satisfiable, then there exists an initial assignment of opinions for which the
\contime is exactly $h+1$. If, however, $\Phi$ is not satisfiable, then any
assignment of opinions will result in a \contime strictly less than $h+1$. We now
give the formal proof.

\subsection*{Reduction}
\begin{figure}[h]
\makebox[\textwidth][c]{
\subcaptionbox{literal nodes and \textsc{or}-gates \label{fig:reduction-or} }{
\includegraphics{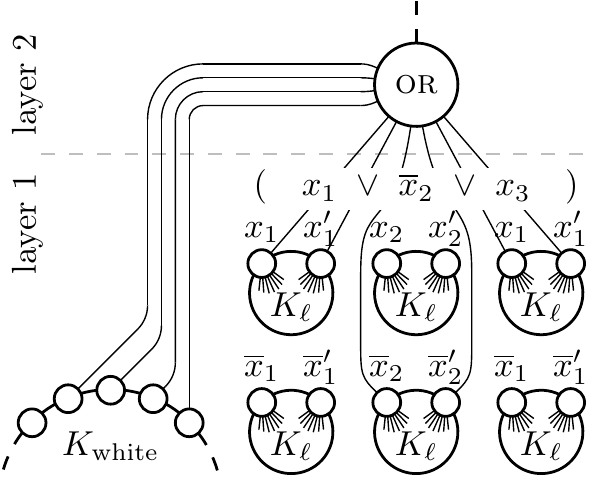}
}
\quad \subcaptionbox{\textsc{and}-gate \label{fig:reduction-and} }{
\includegraphics{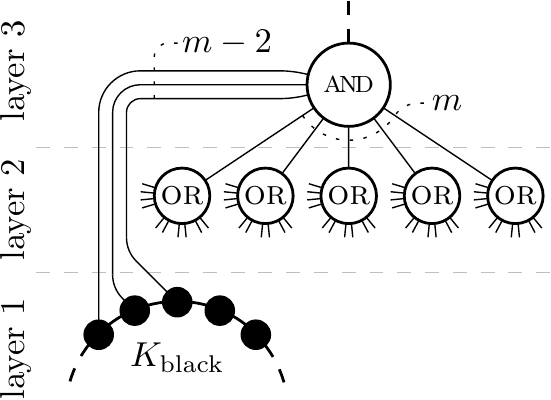}
}
\quad \subcaptionbox{$2/3$-gates \label{fig:reduction-two-third} }{
\includegraphics{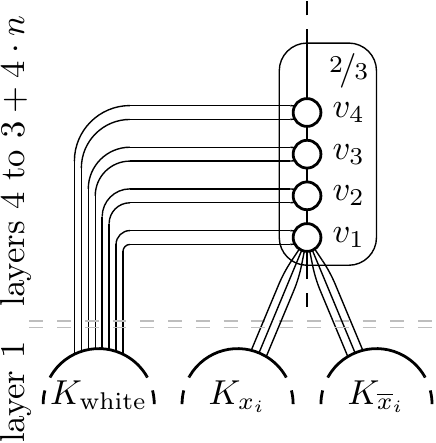}
}
}
\caption{The gates and layers used in the reduction from \textsc{3sat} to \ctdp.}
\end{figure}

Let $\Phi \in $ \textsc{3sat} be a Boolean formula in $3$-conjunctive normal
form. Let $n$ be the number of variables of $\Phi$. Let $m$ be the number of
clauses of $\Phi$. The Boolean formula is of the form
\begin{equation*}
\Phi=(l_{1,1}\vee l_{1,2}\vee l_{1,3})\wedge\cdots\wedge(l_{m,1}\vee l_{m,2}\vee l_{m,3})
\end{equation*}
where $l_{i,j}\in \left\{ x_1, \overline x_1, x_2, \overline x_2,\cdots,x_n,
\overline x_n \right\}$ is a literal for $1\leq i\leq m$ and $1\leq j \leq 3$.

We construct a graph $G$ to simulate the evaluation of $\Phi$ as follows. Let
$\ell = 10\cdot(m + n)+1$. The graph consists of several layers. On the first
layer, we place so-called literal cliques of size $\ell$, and on the layers above we place the
gates. In our reduction, we use \textsc{or}-gates, an \textsc{and}-gate, and
$2/3$-gates. Each gate consists of one or several nodes. Additionally, we have
two so-called \emph{\megaclique{s}} $K_{\text{white}}$ and $K_{\text{black}}$
of size $\ell$.

Let $g$ be an arbitrary but fixed gate. In the following, we will
denote a node on a layer below $g$ that is connected to $g$ as \emph{input
node} to $g$. Additionally, we will denote a node that belongs to $g$ and is
connected to another gate on a layer above $g$ as \emph{output node}.

In the following, we assume that opinion $1$, white, corresponds to Boolean
\textsc{true} and $0$, black, corresponds to \textsc{false}. The main idea of
the construction is to show that an \emph{activation signal} is transmitted
from the bottom up through all layers. If the current assignment of opinions on
the literal cliques corresponds to a satisfying assignment of Boolean values to
$\Phi$, then the process requires $4+4\cdot n$ steps.
The main purpose of the \textsc{or}-gates and the \textsc{and}-gate is to
evaluate $\Phi$. The $2/3$-gates check whether the opinion assignment to
literal nodes is valid. That is, we need to enforce that the corresponding
literal nodes for $x_i$ and $\overline x_i$ are of opposite colors for every
variable $x_i$ of $\Phi$. If either this condition is violated and variables $x_i$
exist for which $x_i = \overline x_i$ or the current assignment of opinions on
the literal cliques does not corresponds to a satisfying assignment of Boolean values to
$\Phi$, the construction enforces that the
process stops prematurely after strictly fewer than $h+1$ steps.

We start by giving a detailed description of the gates and the layers used in our
construction.

\paragraph{Layer $1$ -- Literal Cliques.}
We represent each variable $x_i$ with two cliques, one for $x_i$ and one for
$\overline x_i$. Each clique has a size of $\ell$ which is defined above. Note
that $\ell$ is odd. Additionally, we distinguish three so-called
\emph{representative nodes} in each of these cliques. Furthermore, we add two
cliques of size $\ell$ to the graph which we call \megaclique{s}. Intuitively,
these \megaclique{s} represent the Boolean values \textsc{true} and
\textsc{false}. We will show that they cannot have the same color in order
to achieve a long \contime. The \megaclique{s} are used in all other gates.

\paragraph{Layer $2$ -- Parallel \textsc{or}-Gates.}
The \textsc{or}-gates are placed on layer $2$ and consist of one node $v$ which
is also the output node. 
There is one \textsc{or}-gate for every clause.
Fix a clause $(l_{j,1} \vee l_{j,2} \vee l_{j,3})$.
Input nodes are three pairs of nodes $(v_1$, $v_1')$, $(v_2, v_2')$, and $(v_3,
v_3')$, where $(v_1$, $v_1')$ are two representative nodes of the literal
clique for $l_{j,1}$, $(v_2, v_2')$ are representatives of $l_{j,2}$, and
$(v_3, v_3')$ are representatives of $l_{j,3}$. That is, for each literal in
the clause we connect the \textsc{or}-gate on layer $2$ to two of the three
representative nodes of the corresponding literal clique on layer $1$. The
output node $v$ is additionally connected to $4$ nodes of the
$K_{\text{white}}$ \megaclique. 
Intuitively, we use the \textsc{or}-gates to verify that for each clause at
least one literal is \textsc{true}. All clauses are evaluated simultaneously
using an \textsc{or}-gate for each clause. The \textsc{or}-gate is shown in
\autoref{fig:reduction-or}.

\paragraph{Layer $3$ -- \textsc{and}-Gate.}
There is exactly one \textsc{and}-gate on layer $3$. This \textsc{and}-gate
consists of one output node denoted $u_0$, which has the following input nodes. It is
connected to every output node of the \textsc{or}-gates on layer $2$ and to
$m-2$ distinct nodes of the $K_{\text{black}}$ \megaclique. Intuitively, the
\textsc{and}-gate is used to verify that every clause is satisfied. 

\paragraph{Layers $4$ to $3 + 4\cdot n$ -- $2/3$-Gates.}
The $2/3$-gates consist of a path $v_1$, $v_2$, $v_3$, and $v_4$. Each node of
this path is connected to two distinct nodes of the $K_{\text{white}}$. The
output node of the gate is $v_4$. The node $v_1$ of the first $2/3$-gate on
layer $4$ is connected to the \textsc{and}-gate on layer $3$. The node $v_1$ of
each of the following $2/3$-gates is connected to the node $v_4$ of the previous $2/3$-gate.
Additionally, the input node of the $i$-th $2/3$-gate is connected to three distinct nodes of the
literal clique representing $x_i$ and to three distinct nodes of the literal
clique representing $\overline x_i$ on layer $1$. The output node of the final
$2/3$-gate is connected to $K_{\text{black}}$. This is shown in
\autoref{fig:reduction-two-third}. These gates are used to verify that we do
not have variables $x_i$ in $\Phi$ for which the literal cliques of $x_i$ and $ \overline x_i$ have the same color. Observe that
$2/3$-gates span over $4$ layers and we have $n$ such $2/3$-gates. 

\medskip

Literal cliques, \textsc{or}-gates, and the \textsc{and}-gate use only one
layer, while $2/3$-gates span over $4$ layers. Therefore, the total number of
layers is $h = 3+4\cdot n$, which results from one layer for the literal
cliques, one layer for the \textsc{or}-gates, one layer for the
\textsc{and}-gate, and $4\cdot n$ layers containing $n$ concatenated
$2/3$-gates. 
Based on above description of $G$ we prove the
following lemmas, which are then used to show \autoref{thm:hardness}.

\begin{lemma} \label{lem:claim1}
If $\Phi$ is satisfiable, then there exists an assignment of opinions such that
the \contime in $G$ is at least $h+1$.
\end{lemma}

\begin{proof}
Let $A = (a_1, \dots, a_n)$ be an assignment of Boolean values to the $n$
variables in $\Phi$ which satisfies $\Phi$. We need to show that
there exists an opinion assignment on $G$ for which the \rd requires at least
$h+1$ steps to \conv. In the following, we construct such an opinion assignment.

Let $f_A$ be an initial opinion assignment in the graph $G$ that represents $A$
by initializing the nodes in the literal cliques on layer $1$ according to the
assignment $A$ as follows. For each literal $x_i$ or $\overline x_i$, $a_i$
assigns either \textsc{true} or \textsc{false} to the literal. We denote a
literal $x_i$ or $\overline x_i$ which is assigned \textsc{true} as
\emph{positive} and literals which are assigned \textsc{false} as
\emph{negative}. For the positive literal cliques, we assign the color black to
$\floor{L/2}$ nodes including the representative nodes of the clique. The
remaining $\ceil{L/2}$ nodes, which do not have any other connections except
within the literal clique, are colored white. Negative literal cliques are
colored entirely black. Furthermore, we initialize all nodes of the
$K_{\text{white}}$ and the $K_{\text{black}}$ with white and black,
respectively. All other nodes, the paths $v_1$ to $v_4$ in the $2/3$-gates, the
output nodes of the \textsc{or}-gates, and the output node of the
\textsc{and}-gate, are colored black.

\begin{figure}
\centering
\begin{subfigure}{0.2\textwidth}
\centering
\includegraphics[page=1]{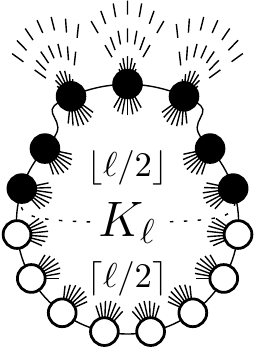}
\caption{$t = 0$}
\end{subfigure}
\begin{subfigure}{0.2\textwidth}
\centering
\includegraphics[page=2]{figures/clique-initialization}
\caption{$t = 1$}
\end{subfigure}
\begin{subfigure}{0.2\textwidth}
\centering
\includegraphics[page=3]{figures/clique-initialization}
\caption{$t = 2$}
\end{subfigure}
\caption{The behavior of the cliques on layer $1$. The three top nodes are the representatives.}
\label{fig:clique-behavior}
\end{figure}

The process now behaves as follows.
\begin{enumerate}
\item In the first step, all black nodes in every positive literal clique
except the representative nodes turn white, since they have $\ceil{L/2}$
white neighbors and only $\floor{L/2}-1$ black neighbors.
\item The representative nodes of the literal cliques will turn white in the
following step. This behavior of the cliques on layer $1$ is shown in
\autoref{fig:clique-behavior}.
\item Additionally to the neighbors
in $K_{\text{white}}$, all \textsc{or}-gates on layer $2$ will have at least two white input
nodes from representing at least one positive literal clique, since $A$
satisfies $\Phi$. Therefore, the \textsc{or}-gates will turn white in step
$3$.
\item Once all \textsc{or}-gates become white, the \textsc{and}-gate has a
total of $m$ white input nodes that form a clear majority against the $m-2$
edges to black nodes in $K_{\text{black}}$ and the $1$ edge to the black node
of the first $2/3$-gate. Therefore, the \textsc{and}-gate turns white in step
$4$.
\item In the following $4\cdot n$ steps, node after node and gate after gate
the $2/3$-gates turn white. Once all nodes of the $2/3$ gates have turned
white, the process stops.
\end{enumerate}
Summing up over all of the above steps, the \contime w.r.t.\ the initial
opinion assignment $f_A$ is exactly $\ctime\left(G(\Phi), f_A\right) = 4 +
4\cdot n = h + 1$. Therefore, the \dtime in $G(\Phi)$ for a satisfiable $\Phi$
is at least $h + 1$, which yields the lemma.
\end{proof}

It remains to show that if $\Phi$ is not satisfiable, then the \dtime
in $G$ is strictly less than $h+1$. Recall that the \dtime is the maximum of the
\contime over all possible initial opinion assignments.
	
\begin{lemma} \label{lem:claim2}
If $\Phi$ is not satisfiable, then there is no assignment of opinions such that
the \contime in $G$ is at least $h+1$.
\end{lemma}

Before we prove this lemma, we establish several auxiliary lemmas which require
the following definitions. Let $u_0$ denote the output node of the
\textsc{and}-gate. Consider the graph $G'$ induced by the nodes of the
\textsc{and}-gate and the nodes of the $2/3$-gates. Let $u_i$ be the node at
distance $i$ to $u_{0}$ in $G'$. That is, $G'$ is a path that consists of the
$4\cdot n+1$ top layers of the graph $G$ and $u_i$ is the $i$-th node on this
path.

\begin{definition}[Stable Time]
We define the \emph{stable time} $s(v)$ for any node $v \in V$ to be the first
time step such that $v$ does not change its opinion in any subsequent time step
$t' > s(v)$ over all possible initial configurations. That is, 
\[ s(v)=\min\left\{ t : \forall f_0 \in \left\{0, 1 \right\}^V ~ \forall t' \geq t \quad f_{t'}(v) = f_t(v) \right\} \enspace . \]
Accordingly, let for any subset $V' \subseteq V$ be $s(V')$ defined as $s(V') =
\max \left\{ s(v) : v \in V'\right\}$.
\end{definition}

\providecommand\Krep{{K^{\text{\textsc{r}}}}}
\providecommand\Kminus{{K^-}}
\providecommand\Kor{{\text{\textsc{or}}}}

In the following, let $V_K$ be the set of nodes of all cliques in $G(\Phi)$,
that is, the nodes contained in the literal cliques and in the \megaclique{s}
on layer $1$. Furthermore, let $V_\Krep$ be the set of representatives of the
cliques and $V_\Kminus = V_K \setminus V_\Krep$. That is, every clique $K$ on
layer $1$ consists of $\Kminus \cup \Krep$. Finally, let $V_\Kor$ be the set of
all output nodes of \textsc{or}-gates. The following lemma shows that the
layers become stable one after the other.



\begin{lemma} \label{lem:stabletimes}
It takes at most $3$ time steps for the layers $1$ and $2$ consisting of
literal cliques and \textsc{or}-gates to become stable. Precisely, we have
\begin{inparaenum}[1)]
\item $s\left(V_\Kminus\right) = 1$,
\item $s\left(V_\Krep\right) = 2$, and
\item $s\left(V_\Kor\right) = 3$.
\end{inparaenum}
\end{lemma}

\begin{proof}
{
\def\itemautorefname~#1\null{(#1)\null}
The lower bounds for all three claims follow from the initial opinion
assignment $f_A$ described in the proof of \autoref{lem:claim1}. We now show
the upper bounds. In the following, let $f_0$ be an arbitrary but fixed initial
opinion assignment.
\begin{enumerate}[{\bfseries (i)}]
\item \label{item:stable-case-1}
Let $K$ be an arbitrary but fixed clique and let $c \in \{ 0,1 \}$ be the
majority color among the nodes of $K$. Let furthermore $\Kminus$ be the set of
clique nodes that do not have connections to any other node except within the
clique, that is, $\Kminus$ contains all clique nodes except representatives.
Note that all nodes in $\Kminus$ only have connections to all other nodes in
$K$. Since $K$ is odd and $c$ is the majority color in $K$, each node $v \in
\Kminus$ with $f_0(v) = c$ will have at least $\floor{\ell/2}$ neighbors out of
a total of $\ell-1$ neighbors colored $c$. Therefore, each node $v \in
\Kminus$ with $f_0(v) = c$ will keep its color $c$ such that $f_1(v) = c$.
However, all other nodes $v' \in \Kminus$ with $f_0(v') \neq c$ will change
their opinion to $c$, since they have at least $\ceil{\ell/2}$ neighbors out of
a total of $\ell-1$ neighbors colored $c$, such that $f_1(v')=c$.

By construction and by the size of the clique, $\ell$, all nodes $v \in
\Kminus$ have more neighbors in $\Kminus$ than in $V \setminus \Kminus$.
Therefore, for all consecutive steps $t' \geq 1$, we have for any $v \in
\Kminus$ that $f_{t'}(v) = c$. This holds for all cliques on layer $1$, and thus
$s\left(V_\Kminus\right) \leq 1$.

\item \label{item:stable-case-2}
Let $K$ be an arbitrary but fixed clique and let $v \in \Krep$ be a
representative node of $K$. By construction, $v$ has a majority of its
neighbors in $\Kminus$ and hence from \autoref{item:stable-case-1} we derive
$s(v) \leq 2$. Therefore, $s\left(V_\Krep\right) \leq 2$. We also observe that
all nodes in $\Krep$ have the same color after the second step, since the nodes
in $\Kminus$ become monochromatic in the first step and these nodes dominate
the behavior of the nodes in $\Krep$.

\item \label{item:stable-case-3}
Let $v$ be the output node of an arbitrary but fixed \textsc{or}-gate in
$V_\Kor$. We observe that all neighbors of $v$ except for one neighbor  (the
node of the \textsc{and}-gate $u_0$) are stable for any time step $t' \geq 2$.
By \autoref{item:stable-case-2}, at time $2$ all representatives of any literal
$x_i$ have the same color and $K_{\text{black}}$ is stable. Therefore, at time
$2$ an even number of neighbors of $v$ are black and an even number is white.
Since the total number of neighbors of $v$ is $10$, we observe that $u_0$
cannot influence $f_{t'}(v)$ for $t'\geq 2$. Moreover, by
\autoref{item:stable-case-1} and \autoref{item:stable-case-2} we have at time
$t' \geq 2$ that the majority of neighbors having color $c$ does not change and
therefore $v$ becomes stable at time $3$. Thus $s\left(V_\Kor\right) \leq 3$
holds. \qedhere
\end{enumerate}
}\end{proof}

The above lemma gives  bounds on the stable time of layers $1$ and $2$. In
the following, we argue that whenever a node changes its opinion in any step
$t$ after time step $3$, it will not change its color in any subsequent time
step $t' \geq t$ any more. We therefore define the so-called \emph{activation
time} of a node $v \in G'$ as follows.

\begin{definition}[Activation Time]
Let $c$ be the color of the $K_{\text{black}}$ \megaclique at time $2$ and let
$f_0$ be an arbitrary but fixed initial opinion assignment. We define the
\emph{activation time} of a node $v \in G'$ to be the first time step after
time step $3$ in which the node $v$ adopts opinion $c$. That is,
$a(v) = \min\left\{t \geq 3 : f_t(v)=c\right\}$.
If $v$ does not change its color after time step $3$ we write $a(v)=3$.
\end{definition}

We now use above definition to state the following lemma, which describes that
every node $u_i \in G'$ with $i \geq 1$ changes its color at most once after
time step $3$. Note that this covers the nodes of the $2/3$-gates.

\begin{lemma}\label{obs:onlyonce}
Let $f_0$ be an arbitrary but fixed initial opinion assignment. Let 
$t$ be the activation time w.r.t.\ $f_0$ of the node $u_i \in G'$ with $i \geq
1$ such that $t = a(u_i)$. Then for all $t' \geq t$ we have
$f_{t'}(u_i)=f_{t}(u_i)$.
\end{lemma}
\begin{proof}
By \autoref{lem:stabletimes}, all nodes $u \in V_\Krep$ are stable at $t' \geq
2$. We now distinguish two cases.
\paragraph{Case 1: $i \mod 4 \neq 1$.}
Observe that $u_i$ can only change its color at time $t = a(u_i)$, if it had a
different color than $K_{\text{white}}$ in the previous round. This holds,
since every node $u_i$ with $i \mod 4 \neq 1$ has the same number of
connections to $K_{\text{white}}$ than to nodes in $V \setminus
K_{\text{white}}$. Since furthermore the process behaves lazy, any node $u_i$
which has the same color as $K_{\text{white}}$ cannot change its opinion back to
the opposite color any more.
\paragraph{Case 2: $i \mod 4 = 1$.}
The node $u_i$ is a $v_1$ node of the $j$-th $2/3$-gate with $j = \ceil{i/4}$.
Therefore it is connected to three representatives of each literal clique for
$x_j$ and $\overline x_j$. The literal representatives of $x_j$ and $\overline
x_j$ are stable at time $t' \geq 2$. Now if $x_j$ and $\overline x_j$ have the
same color $c$, then $u_i$ has $6 > |N(u_i)|/2$ edges to nodes of color $c$.
Therefore, the node does not change its color any more after time step $3$.
That is, we have $a(u_i)=3$ and also $f_{t'}(u_i)=c$ for any consecutive time
step $t' \geq 3$. If, however, $x_j$ and $\overline x_j$ do not have the same
color, these edge \emph{cancel} each other out and the color of node $u_i$ is
determined by $u_{i-1}$, $u_{i+1}$, and $K_{\text{white}}$. Therefore, the same
argument as in the first case holds. 
\end{proof}

In the following we examine the behavior of layer $3$ which contains only of
the \textsc{and}-gate. Recall that $u_0$ is the output node of the
\textsc{and}-gate. The next lemma describes the following fact. The
\textsc{and}-gate $u_0$ can only change its color in a time step $t \geq 4$ if
$u_1$ changed its color in time step $t-1$. After this change at time $t$, the
node $u_0$ cannot change its color again.

\begin{lemma}\label{lem:stable-and}
Let $f_0$ be an arbitrary but fixed initial opinion assignment and let
furthermore $t$ be the round after node $u_1$ has been activated such that $t =
a(u_1) + 1$. For all consecutive rounds $t' \geq t$ we have $f_{t'}(u_0) =
f_{t}(u_0)$. That is, the \textsc{and}-gate does not change its opinion any
more once the node $u_1$ has become stable.
{\normalfont \hfill (see \autoref{apx:np-hardness})}

%
%
\end{lemma}

The following lemma implies that in order to reach a \contime of $h+1$ the
gates on the path $u_0, \dots, u_k$ in $G'$ have to activate one after the
other starting with $u_0$ at time $4$. Recall that $k = 4 \cdot n$.

\begin{lemma}\label{lem:acttooearly}
Let $f_0$ be an arbitrary but fixed initial opinion assignment and let $u_i \in
G'$ be a node with $0\leq i \leq k$. If $a(u_i)< i+4$ w.r.t.\ $f_0$, then
$\ctime\left(G(\Phi), f_0\right) < h+1$.
{\normalfont \hfill (see \autoref{apx:np-hardness})}
\end{lemma}

In the following two lemmas, we enforce that initial opinion assignments which
do not represent valid assignments of Boolean values to literal cliques result
in premature termination of the \rd in $G(\Phi)$. An assignment is called
\emph{illegal} if there exist literal cliques such that the majority of $x_i$
and the majority of $\overline x_i$ have the same initial color.

\begin{lemma}\label{lem:illegalass}
Let $f_I$ be an illegal initial opinion assignment. The \contime
$\ctime\left(G(\Phi), f_I\right)$ is strictly less than $h+1$.
{\normalfont \hfill (see \autoref{apx:np-hardness})}
\end{lemma}

\begin{lemma} \label{lem:black-white-clique}
If after two time steps $K_{\text{white}}$ and $K_{\text{black}}$ have the same
color, the process stops after strictly fewer steps than $h+1$.
\end{lemma}

\begin{proof}
Let $c$ be the color of both \megaclique{s} after two time steps. Note that
from \autoref{lem:stabletimes} we conclude that both cliques are stable at time
$2$. Therefore $u_{k}$ activates at most at time $3$, that is, $a(u_k)=3$. By
induction, one can show that $u_i$ will activate at most at time $3+k-i$. Hence
$u_1$ becomes activated at most at time $t=3+k-1<h$ and $u_0$ at most at time
$t = 3 + k$ which is strictly less than $h+1$. Since by
\autoref{lem:stabletimes} all other nodes are also stable at time $3+k<h+1$ the
claim follows.
\end{proof}

We now combine above lemmas and prove \autoref{lem:claim2}.

\begin{proof}[Proof of \autoref{lem:claim2}]
In the following we assume that $K_{\text{white}}$ and $K_{\text{black}}$ have
opposite colors after the second step, since otherwise the \contime is less
than $h+1$ as shown in \autoref{lem:black-white-clique}. W.l.o.g., assume
$K_{\text{white}}$ is colored white and $K_{\text{black}}$ is colored black.
Furthermore, we assume that the assignment is legal, since otherwise the
\contime is less than $h+1$ as shown in \autoref{lem:illegalass}. Finally, we
also assume that $u_1, \dots, u_k$ are initially black, since otherwise the
\contime is less than $h+1$ as shown in \autoref{lem:acttooearly}. Note that
this especially covers the node $u_1$ which we assume to be black at time $4$,
since otherwise again the \contime is less than $h+1$ according to
\autoref{lem:acttooearly}. 

According to the assumption of \autoref{lem:claim2}, $\Phi$ is not satisfiable.
That is, for every possible assignment of Boolean values to the variables in
$\Phi$, there exists a clause $(l_1 \vee l_2 \vee l_3)$ where all literals
$l_1$, $l_2$, and $l_3$ are \textsc{false}. Therefore, for any legal initial
opinion assignment $f_0$ in $G(\Phi)$, the representative nodes of the
corresponding literal cliques will be black at time $2$. Consequently, the
\textsc{or}-gate corresponding to that clause will be stable with color black
at time $3$.

This implies that the \textsc{and}-gate is black as long as $u_4$ is black
since at least $(m-2)+1+1 > |N(u_0)|/2$ neighbors are black. Since the
\textsc{and}-gate is black, we can only have $a(u_1)=5$ if $a(u_2)=4$.
According to \autoref{lem:acttooearly}, this results in a \contime strictly
less than $h+1$. Note that if $f_3(u_1)=1$, then $u_2$ will be activated at
time $4$ and again by \autoref{lem:acttooearly} this yields that the \contime
is less than $h+1$.
\end{proof}

Now we combine \autoref{lem:claim1} and \autoref{lem:claim2} to  show
\autoref{thm:hardness}.

\begin{proof}[Proof of \autoref{thm:hardness}]
It is easy to see that \ctdp is in NP. Furthermore, we can polynomially reduce
3\textsc{sat} to \ctdp. The correctness proof of the reduction follows from
\autoref{lem:claim1} and \autoref{lem:claim2}. Therefore, \ctdp is NP-complete.
\end{proof}

%


\section{Bounds on the \Dtime} \label{sect:main-result}
Since the problem is NP hard, we cannot hope
to calculate the voting time of a graph efficiently. Nevertheless, in this section we show, that it is possible to obtain non-trivial upper bounds on the \dtime
that are easy to compute. This section is dedicated to proving our upper bound on the \dtime.\autoref{thm:asym_upper}. The main contribution of this theorem  is the influence of 
symmetry which is studied in \autoref{sect:symmetry}.

We start by giving a formal version of the potential function argument \cite{GO80, PS83} as conceived in \cite{Win08}. In
the following we assume that each edge in $\{x,y\} \in E$ can be replaced by
two directed edges $(x,y)$ and $(y,x)$. The main idea is based on so-called
\emph{bad arrows} defined as follows.

\begin{definition} \label{def:bad-arrow}
Let $G = (V, E)$ be a graph with initial opinion assignment $f_0$. Let $v$
denote an arbitrary but fixed node and $u \in N(v)$ a neighbor of $v$. Let $t$
denote an arbitrary but fixed round. The directed edge $(v, u)$ is called
\emph{bad arrow} if and only if the opinion of $u$ in round $t+1$ differs from
the opinion of $v$ in round $t$.
\end{definition}

\begin{wrapfigure}{r}{0.3\textwidth} \centering
\vspace{-1\baselineskip}
\includegraphics{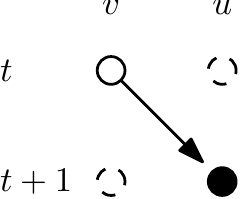}
{
\captionsetup{format=plain,font=small}
\caption{A \emph{bad arrow} from node $v$ to node $u$ in round~$t$.}
\label{fig:bad-arrow}
}
\vspace{-1\baselineskip}
\end{wrapfigure}

Intuitively, each of these directed edges $(v, u)$ can be seen as \emph{advice} given
from $v$ to $u$ in the voting process. In the case of a bad arrow the advice
was not followed by $u$ since it has a different opinion in the following round
than $v$. Observe that each bad arrow is incident at exactly two nodes and
thus we say it is \emph{outgoing} in the node at its tail and incoming in the
node at its head.
An example of such a bad arrow can be seen in \autoref{fig:bad-arrow}.

\begin{theorem} \label{thm:Winkler}
Let $G = (V, E)$ be a graph which contains only vertices of odd
degree.  The \dtime of the \rd on $G$ is at most $1+W$ where $W$ is an 
upper bound on the initial number of bad arrows for any initial opinion assignment on $G$.
In particular, the \dtime of $G$ is at most $2\cdot|E| + 1$.
\end{theorem}

For a proof, see \autoref{apx:main-result}. Note that
in \autoref{thm:Winkler} it is assumed that all nodes of the graph have odd-degree.
In the following we show how to remove this assumption. 
Let in the following \Veven be the set of even-degree vertices
in $V$ and, analogously, \Vodd be the set of odd-degree vertices. Clearly,
$V = \Veven \cup \Vodd$.

\begin{definition}\label{def:G-star}
Let $G=\left( V,E \right)$ be a graph. The graph $G^*=\left(V,E^*\right)$ is the graph obtained by adding a self loop to every node of even degree in $G$. More formally,
\begin{equation*}
E^* = E ~\cup \bigcup_{v \in \Veven} \left(v, v\right) \enspace .
\end{equation*}
\end{definition}
From the definition it follows that
$|E^*|=|E|+|\Veven|$.

\begin{theorem}\label{thm:2E_upper}
The \dtime of the \rd on any graph $G=(V,E)$ is
at most $1+W_{bad}$, where $W_{bad}$ is the number of
initial bad arrows in $G^*$.
\end{theorem}

For a proof, see \autoref{apx:main-result}.
The upper bound on the \dtime considered in \cite{KPW14} follows 
from the $2\cdot|E|$ upper bound on the number of bad arrows of \autoref{thm:Winkler}.  
In the following we show that this result
can be improved further by a factor of $2$ by simply applying the following
lemma. 

\begin{lemma} \label{lem:max_bad_arrows}
Let $G$ be a graph. 
The number of initial bad arrows in $G^*$ 
can be at most $|E|-\ifrac{|V_{odd}|}{2}$.
\end{lemma}

Therefore, combining \autoref{thm:2E_upper} and \autoref{lem:max_bad_arrows} we obtain the following corollary.
\begin{corollary} \label{thm:E_upper}
The \dtime of the \rd on any graph $G=(V,E)$ is
at most $1+|E|-\ifrac{|V_{odd}|}{2}$.
\end{corollary}
%

Note that \autoref{thm:E_upper} is tight for general graphs up to an additive
constant of~$1$.
Indeed, consider a path with an initial opinion assignment on which the
opinions alternate except for the last two nodes, which share the same opinion.

\subsection{Improved Bounds for Dense Graphs}
\label{sect:improved-bounds}

We observe that \autoref{thm:E_upper}
is (almost) tight, and it gives us a \dtime linear in the number of vertices for sparse graphs where $|E| =
\BigO{|V|}$. However, for dense graphs with $|E| = \BigOmega{|V|^2}$ there is
room for improvement. Now the main goal in this following subsection is to
reduce the dominant term of the \dtime even further, which leads us
to the following theorem.

\begin{theorem} \label{thm:half_E_upper}
Let $G=(V, E)$ denote a graph. For any initial opinion assignment $f_0$ on $G$,
the \dtime of the \rd is at most
$1+\frac{|E|}{2}+\frac{|\Veven|}{4}+\frac{7}{4}\cdot|V|$.
{\normalfont \hfill (see \autoref{apx:main-result})}
\end{theorem}

\subsection{The Influence of Symmetry} \label{sect:symmetry}

We observe that the \rd is much faster on graphs that exhibit certain types
of symmetry, such as the star graph, the complete graph
and many other graphs in which several nodes share a common neighborhood.
We investigate this feature of the process to further improve the bounds 
obtained so far. 
We recall that a set of nodes $S$ is called a \textit{family} if and only if for all nodes $u,v\in S$ we have
$N\left( u \right)\setminus\left\{ v \right\}=N\left( v \right)\setminus\left\{ u \right\}$.
The key fact is that these nodes of any family will
behave in a similar way after the first step.
\begin{definition} \label{def:asym-graph}
Let $\family{u}$ denote the family $u$ belongs to.
We write $u \sim v$ if $\family{u}=\family{v}$.
\end{definition}

\begin{lemma}  \label{lem:transitivity}
The relation $\sim$ defines an equivalence class. In particular, all nodes in the same family either form a 
clique or are all pairwise non-adjacent, and they all have the same degree in $G$.
{\normalfont \hfill (see \autoref{apx:main-result})}
\end{lemma}

\begin{corollary}
 For any graph $G$, its asymmetric graph $G^\Delta$ is 
well-defined.
\end{corollary}
\begin{proof}
Thanks to \autoref{lem:transitivity}, the set of families is a partition of the nodes of $G$.
By construction of $G^\Delta$, every family $S$ in $G$ is replaced by one or two nodes in $G^\Delta$.  
Therefore, there is a bijection between the families in $G$ and the corresponding node or pair of nodes in $G^\Delta$. Hence $G^\Delta$ is well-defined.\end{proof}

We now prove \autoref{thm:asym_upper}.

\begin{proof}[Proof of \autoref{thm:asym_upper}]
Let $v$ and $v'$ be two nodes of the same family $\family{v}=\family{v'}$, having the same color at time $t$. 
Since $v$ and $v'$ observe the same opinions in their respective neighborhood,
$v$ and $v'$ will also have the same color anytime after $t$.
It follows that if at some time $t$ there is a bad arrow going from $v$ to
some neighbor $u$ (or from $u$ to $v$), 
then there will also be a bad arrow from $v'$ to $u$ (or from $u$ to $v'$).
In particular, this implies that whenever the number of bad arrows adjacent
to $v$ is decreased by some amount $c$, also the identical number of bad arrows
adjacent to $v'$ will be decrease by the same amount $c$.

Now recall the proofs of \autoref{thm:E_upper} and \autoref{thm:half_E_upper}.
An estimate of the \dtime is obtained by upper bounding the number
of bad arrows that can possibly disappear during the process.
The main argument is the following. It suffices to only consider the bad-arrows adjacent to $v$ in $G^\Delta$, since the corresponding bad arrows adjacent to $v'$ will disappear whenever 
those adjacent to $v$ do.

Now let $v$ and $v'$ be two nodes with $\family{v}=\family{v'}$ having a different color at time $t$.
We can divide every such family that contains nodes of different opinions into two sets $S_0$ and $S_1$ according to their initial opinion in the first round.
Note that all nodes in either set behave identically.
In particular, an adjacent bad arrow from a node $u$ to all nodes of either set disappears at the same time. Since there is bijection between the families of $G$ and the pairs of nodes and singletons of $G^\Delta$,
and  by applying 
\autoref{thm:E_upper} and \autoref{thm:half_E_upper}
we can bound the \dtime by bounding the bad arrows in $G^\Delta$.
This yields the first part of the claim. Using \cite{CH94}, one can obtain the
modular composition of $G$ in $\BigO{|E|}$ time steps. In another $\BigO{|E|}$ time
steps one can select from the modular composition those modules that form a
family, using  that all nodes of a family have the same degree. Hence,
$G^\Delta$ can be constructed in linear time.\end{proof}

\vfill
\paragraph{Acknowledgments.}
We would like to thank our supervisors Petra Berenbrink, Andrea Clementi, and Robert Elsässer for helpful discussions and important hints.
\newpage
\bibliographystyle{alpha}
\bibliography{paper}
\newpage
\appendix
\section*{\textsc{Appendix}}
\section{Omitted Proofs from \autoref{sec:np}} \label{apx:np-hardness}

\begin{proof}[Proof of \autoref{lem:stable-and}]
Note that $t$ is at least $4$ by definition of the activation time. Let $c$ be
the color of $K_{\text{black}}$ and $\overline c = 1-c$ the opposite color of
$c$. If at most $m-2$ of the \textsc{or}-gates have color $\overline c$, then
the node of the \textsc{and}-gate has at least $2+(m-2)>|N(u_0)|/2$ neighbors
which will be colored $ c$ for all $t\geq 3$ and therefore the
\textsc{and}-gate will be colored $ c$ for every $t'\geq 4$.

If, however, $m-1$ of the \textsc{or}-gates have color $\overline c$, only one
\textsc{or}-gate has not been activated and has color $c$. Thus the node of the
\textsc{and}-gate $u_0$ has on layer $1$ and layer $2$ a total of $m-1$
neighbors of color $c$ and also a total of $m-1$ neighbors of color $\overline
c$. That is, these neighbors \emph{cancel} each other out. By
\autoref{lem:stabletimes} the cliques and gates on layers $1$ and $2$ do not
change their color for any $t'\geq 4$. Therefore, the node $u_0$ can only be
influenced by $u_1$ and the color of $u_0$ at time $t$ is the color of $u_1$ at
time $t-1$ for any $t\geq 4$. By \autoref{obs:onlyonce} we know that $u_1$ may
change its opinion only once in a round $t = a(u_1) \geq 3$ and therefore for
any round $t' \geq t + 1$ we have $f_{t'}(u_0)=f_{t}(u_0)$.

Finally, if $m$ of the \textsc{or}-gates are colored $\overline c$, then $u_0$
has $m>|N(u_0)|/2$ neighbors of color $\overline c$ and since by
\autoref{lem:stabletimes} these $m$ neighbors do not change their color for
$t\geq 4$ we have $f_t(u_0)=\overline c$ for all $t\geq 4$. Thus, in all cases
the claim follows.
\end{proof}

\begin{proof}[Proof of \autoref{lem:acttooearly}]
By \autoref{lem:stabletimes} all nodes of $V_K$ and $V_\Kor$ are stable after
time step $2$ and $3$, respectively. From \autoref{obs:onlyonce} we observe
that every node of $u_1, \dots, u_k$ with $k=4\cdot n$ can only change its
color once after time step $3$. Note that from \autoref{lem:stable-and} we
conclude that if $u_1$ changes its color at time $t$ then the \textsc{and}-gate
does not change its color for any $t'\geq t+1$. 

We now consider the \emph{inner} nodes of the path $u_j$ for which $1\leq j<k$.
In order for a node $u_j$ to change its color at time $t > 3$, one of the
neighboring nodes $u_{j-1}$ or $u_{j+1}$ must have changed its color at time
$t-1$. This follows, since according to \autoref{lem:stabletimes} all other
neighbors of the node $u_j$ are already stable after $2$ steps. Now if a node
$u_j$ changes its opinion, one of the neighbors of $u_j$ must have changed its
opinion in the previous round. This can only be either $u_{j-1}$ or $u_{j+1}$
(or both), since all other neighbors of $u_j$ are already stable.

Since all nodes $u_1, \dots, u_k$ of the path in $G'$ can only change their
color once and since $u_0$ becomes stable one time step after $u_1$ changes its color, the
\contime of the graph $G(\Phi)$ is dominated by the behavior of the path. That
is, in order to achieve a \emph{long} \contime, the path must change
its color one node after the other, resulting in a \contime in $\BigOmega{n}$.
Observe that this can only happen if the entire path has a different color
than the $K_\text{white}$ after the second step. As soon as one of the path
nodes is assigned the same opinion as the $K_\text{white}$ \megaclique, the
entire path will be activated too early and the process stops prematurely.

Now in order to have a \contime of $h+1$, the path in $G'$, $u_0, \dots, u_k$,
must activate from $u_0$ over $u_1$ up to $u_k$ or in the reverse direction
from $u_k$ over $u_{k-1}$ down to $u_0$. We now argue that activating from $u_k$
down to $u_0$ cannot yield a \contime of at least $h+1$. 

Note that all neighbors of $u_k$ except for $u_{k-1}$ are stable at any time
step $t \geq 3$. Therefore, $u_k$ either has the same fixed opinion as the
$K_{\text{white}}$ and thus $a(u_k) = 3$, or $u_k$ has an activation time
$a(u_k) = a(u_{k-1}) + 1$. Now in the first case, $a(u_k) = 3$, the \contime is
bounded by $3+k=3+4n < h+1$, since the path becomes stable one node after the
other starting with the node $u_k$. That is, the resulting \contime is strictly
less than $h+1$. In the second case, $a(u_k) \geq 4$, we note that $a(u_k) =
a(u_{k-1}) + 1$ and thus the path cannot activate from $u_k$ down to $u_0$.

We conclude that in order to have a \contime of $h+1$ the nodes must activate
from $u_0$ to $u_k$ starting with node $u_0$ in time step $4$ such that
$a(u_0)=4$. Therefore, $a(u_i)$ must be $i + 4$ to have a \contime of
$h+1$ which shows the lemma.
\end{proof}

\begin{proof}[Proof of \autoref{lem:illegalass}]
In the following we use $\Krep(x_i)$ and $\Krep(\overline x_i)$ to denote the
representative nodes of the literal cliques for $x_i$ and $\overline x_i$.
Note that by \autoref{lem:stabletimes} these representative nodes are stable at
time $2$. Now assume both cliques have color $c$ after the second step.

Let $u$ be the first node $v_1$ of the $i$-th $2/3$-gate. By the construction
of $G(\Phi)$, $u$ is connected to $6$ representative nodes of literal cliques
which all share the same color $c$. Since the representative nodes are stable
after $2$ steps, $u$ will also have color $c$ for every time step $t' \geq 3$.
That is, $a(u) = 3$ and thus by \autoref{lem:acttooearly} the \contime is less
than $h+1$.
\end{proof}

\section{Omitted Proofs from \autoref{sect:main-result}} \label{apx:main-result}

\begin{proof}[Proof of \autoref{thm:Winkler}]

The idea of the proof is to define a potential
function $\phi_t$ that is strictly monotonically decreasing over the time. 
Let $f_0$ be any initial opinion assignment. The potential function $\phi_t$ is simply the number of
bad arrows defined in \autoref{def:bad-arrow}, that is
\begin{equation*}
\phi_t = \phi_t(G, f_t) = \left|\left\{ (v, u) \in E: f_{t+1}(u) \neq f_{t}(v) \right\}\right| \enspace .
\end{equation*}

Let $v$ denote an arbitrary but fixed node. To show that $\phi_t$ indeed
is a strictly monotonically decreasing potential function as long as $t \leq
\ctime(G, f_0)$ we distinguish the following two cases.

\noindent \textbf{Case 1.} The node $v$ has the same opinion in round $t+1$ as in round $t-1$, that is, $f_{t+1}(v) = f_{t-1}(v)$.

For each neighbor $u$ of $v$ that has a different opinion in round $t$ than
$v$ in round $t-1$, there is a bad arrow from $v$ to $u$. We denote the number
of these outgoing bad arrows leaving round $t-1$ as $m_{t-1}(v)$, that is
\[m_{t-1}(v):=\left|\left\{u\in N(v)\,\middle |\,f_{t}(u) \neq f_{t-1}(v)\right\}\right| \enspace .\]
  There is an incoming bad arrow at node $v$ in round $t+1$ from each
neighbor that has a different opinion in round $t$. 
 Let $n_{t+1}(v)$ be this number, that is
\[n_{t+1}(v):=\left|\left\{u\in N(v)\,\middle |\,f_{t}(u) \neq f_{t+1}(v)\right\}\right| \enspace .\]
Now recall that $v$ has the same opinion in
round $t+1$ as in round $t-1$. Thus, the number of incoming bad arrows at node $v$ in
round $t+1$ is the same as the number of bad arrows leaving node $v$ in
round $t-1$, which gives us 
\begin{equation}
n_{t+1}(v) = m_{t-1}(v) \enspace .\label{eq:2period_case1}
\end{equation} An example for this case is shown
in \autoref{fig:winkler-case-1}.

\noindent \textbf{Case 2.} The node $v$ has a different opinion in round $t+1$
than in round $t-1$, that is, $f_{t+1}(v) \neq f_{t-1}(v)$.

Let $m_{t-1}(v)$ and $n_{t+1}(v)$ be defined as above.
Since $v$ changed its opinion after round $t-1$,
either in step $t$ or in step $t+1$, there is an incoming bad
arrow at node $v$ in round $t+1$ for every neighbor of $v$ that did not have an
incoming bad arrow in round $t$. Now the
key is that node $v$ can only have its current opinion in round $t+1$ if there
is a clear majority in round $t$ in favor of this opinion among all of its
neighbors. Observe that this is where the odd degrees mentioned in the
problem statement \cite{Win08b} indeed play a role. Since every node has odd
degree, there is always a clear majority among its neighbors and no tie
between opinions can ever occur. Now if there is a clear majority in round $t$,
the number of incoming bad arrows at node $v$ in round $t+1$ will be strictly
smaller than the number of outgoing bad arrows at node $v$ in round $t-1$, that is
\begin{equation}
n_{t+1}(v) < m_{t-1}(v) \enspace .\label{eq:2period_case2}
\end{equation}  An example for this case is shown in
\autoref{fig:winkler-case-2}.

\begin{figure} \centering
\begin{subfigure}{0.45\textwidth} \centering
\includegraphics{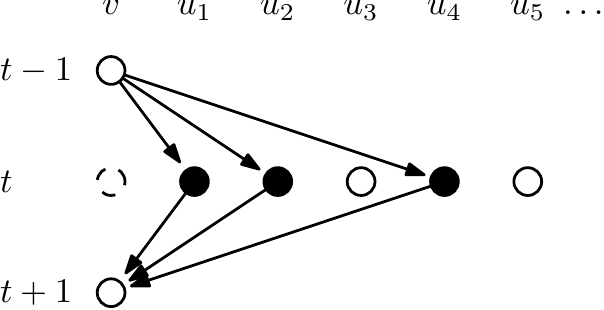}
\caption{Case 1: $f_{t+1}(v) = f_{t-1}(v)$}
\label{fig:winkler-case-1}
\end{subfigure}
\begin{subfigure}{0.45\textwidth} \centering
\includegraphics{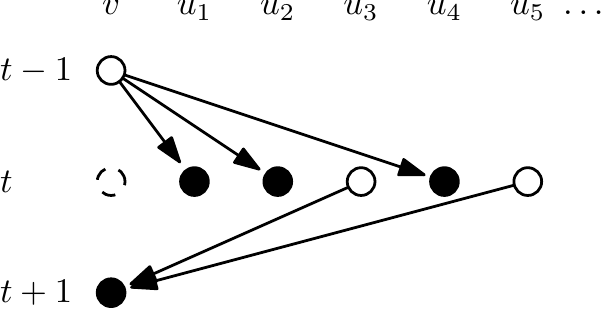}
\caption{Case 2: $f_{t+1}(v) \neq f_{t-1}(v)$}
\label{fig:winkler-case-2}
\end{subfigure}
\caption{Above figures show examples for the two cases. In the first case, the
number of outgoing bad arrows from $v$ in round $t-1$ equals the number of
incoming bad arrows at $v$ in round $t+1$. In the second case the node $v$ has
color black in round $t+1$ due to a majority for black in round $t$. Therefore
the number of incoming bad arrows at $v$ in round $t+1$ is strictly smaller
than the number of outgoing bad arrows from node $v$ in round $t-1$.}
\end{figure}

\noindent \textbf{Both cases.}
We take the sum over all outgoing bad arrows leaving the nodes in round $t-1$ and obtain
$M_{t-1} = \sum_{v \in V} m_{t-1}(v)$. Analogously, we take the sum 
over all incoming bad arrows in round $t+1$
which gives us $N_{t+1} = \sum_{v \in V} n_{t+1}(v)$.
However, since each bad arrow is incident in exactly two nodes, we conclude that
the sum over all incoming bad arrows in round $t+1$ is the same as the sum over
all outgoing bad arrows in round $t$. This gives us
\begin{equation}
M_{t} = N_{t+1} = \phi_{t} \enspace .\label{eq:2period_num-of-bad}
\end{equation}

If the \rd has reached a \stab
state in round $t$, from \autoref{eq:2period_case1} and \autoref{eq:2period_num-of-bad} we get
\begin{equation*}
 \phi_{t} = N_{t+1} = M_{t} =  \sum_{v \in V} m_{t}(v) = \sum_{v \in V} n_{t+2}(v) = N_{t+2} =  \phi_{t+1} \enspace .
\end{equation*}

Now assume that the \rd has not yet reached a \stab
state in round $t$. That is, at least one node has a different opinion in round
$t+1$ than it had in round $t-1$. Then from \autoref{eq:2period_case2} and \autoref{eq:2period_num-of-bad} we get
\begin{equation*}
 \phi_{t} = N_{t+1} = M_{t} =  \sum_{v \in V} m_{t}(v) > \sum_{v \in V} n_{t+2}(v) = N_{t+2} =  \phi_{t+1}
\end{equation*}
which proves that the \dtime of the \rd on $G$ is bounded from above by the initial number of bad arrows.

In particular, since there can be a bad arrow only between ordered pairs of adjacent nodes, the initial number of bad arrows is bounded by
$2\cdot|E|$. Together with the observation that above argument can only be applied after the first step this implies that
\begin{equation*}
\ctime \leq 2\cdot|E| + 1 \enspace . \qedhere
\end{equation*}
\end{proof}

\begin{proof}[Proof of \autoref{thm:2E_upper}]
For every node $v \in V$ the sequence of opinions,
$\left(f_t(v)\right)$, is exactly the same for the \rd in
$G$ as in the \rd in $G^{*}$. Indeed, every odd-degree node
has the same neighborhood in both, $G$ and $G^*$, thus the process
is the same for these nodes. Now consider an arbitrary
even-degree node $v$ and fix a round $t$. If in $G$ there is a tie in round $t$, $v$ behaves lazily
in $G$ and keeps its own opinion at round $t$. In $G^*$, the node $v$ considers its own
opinion and thus also stays with its own opinion. If
on the other hand there is a clear majority in $G$, this majority has a winning
margin of at least $2$, since $v$ has even degree. Thus, the impact of the self loop
can be neglected and again $v$ behaves the same in $G^*$ as in $G$.

We can thus bound the \dtime of $G$ by applying \autoref{thm:Winkler} to the odd-degree graph $G^{*}$. 
\end{proof}

\begin{proof}[Proof of \autoref{lem:max_bad_arrows}]
From the definition of the \rd we conclude that only
less than half of a node's neighbors could have had a different
opinion at time $t=0$, since otherwise the node would have changed its own
opinion. Formally, for any $v\in V$ it holds that
\[
	\sum_{u\in N\left( v \right)}\left[ f_1\left( v \right)\neq 
		f_0\left( u \right) \right] \leq \frac{ |N\left( v \right)| }2 \enspace .
\]
Also, for odd-degree nodes the above inequality is strict. 
Therefore, the number of incoming bad arrows at a node at time $t=1$ is
smaller than half of its degree (strictly, for odd nodes). 
Thus, summing up all initial bad arrows we get 
$\ifrac{\left(2\cdot|E|-|\Vodd|\right)}{2}$, 
which concludes the proof.
\end{proof}

\subsection{Omitted Proofs from \autoref{sect:improved-bounds}}
\label{apx:improved-bounds}

\begin{definition} \label{def:qswap}
In an arbitrary round $t$ an opinion assignment $f'_t$ is a \emph{\qswap} of $f_t$ if for all nodes $v$
\begin{equation*}
	f'_t(v) = f_t(v) \quad \vee \quad f'_t(v) = q\enspace .
\end{equation*}
That is, all opinions assigned by $f'_t$ are either the original opinion assigned by $f_t$ or $q$.
\end{definition}

Based on this definition we can state and prove the following key lemma.

\begin{lemma}[Monotonicity] \label{lem:monotone}
Let $f_t$ be an opinion assignment in round $t$ and $f'_t$ a \qswap of $f_t$.
Let furthermore $v$ be a node for which $f_t(v) \neq f'_t(v)$. It holds for any
time step $k \geq t$ that
\begin{equation*}
f_k(v) = q \implies f'_k(v) = q \enspace .
\end{equation*}
Furthermore, any subsequent opinion assignment $f'_k$ is a \qswap of $f_k$.
\end{lemma}
\begin{proof}
We show \autoref{lem:monotone} by induction over $k$. The base case for $k = t$
is trivially true. Now suppose that \autoref{lem:monotone} holds for $k \leq
m$. Let $v$ be an arbitrary but fixed node for which $f_{m+1}(v) = q$. Since
$f_{m+1}(v) = q$ we had a majority for $q$ among the neighbors of $v$ in the
previous opinion assignment $f_{m}$ and according to the induction hypothesis
$f'_m$ is a \qswap of $f_m$. Therefore, in $f'_{m}$ the number of nodes with
opinion $q$ could have only increased, strengthening the majority for opinion
$q$ even further. Thus, $f'_{m+1}(v) = q$ holds. Now assume $f'_{m+1}$ was not
a \qswap of $f_{m+1}$. That is, there exists a node $u$ for which $f'_{m+1}(u)
\neq f_{m+1}(u)$ and $f'_{m+1}(u) \neq q$. This is a contradiction to the
previous statement. Together, this concludes the induction.
\end{proof}

In other words, \autoref{lem:monotone} states that \emph{strengthening} an
opinion will never make it weaker in a subsequent round, that is, if a node
ends up with opinion $q$, it also ends up with the same opinion in the
\qswap{}ped opinion assignment.

\begin{definition}
An opinion assignment $f_t$ is \emph{\qfund} if $f_{t+2}$ is a \qswap of $f_t$.
\end{definition}

We now use the definition above to further bound the \dtime, since
the \rd has the property that once the process is either in a \fund{0} or
\fund{1} state it will \stabil in a number of steps linear in $|V|$.
Furthermore, note that a \sstate is both \fund{0} and
\fund{1}.

\begin{lemma} \label{lem:fundamentalist_stabilization}
A \qfund opinion assignment $f_t$ \stabil{s} in time $2\cdot|\{v \in V: f_t(v) \neq q\}|$.
\end{lemma}

\begin{proof}
By definition $f_{t+2}$ is a \qswap of $f_t$. Thus we can apply
\autoref{lem:monotone} and conclude that either all nodes have at time $t+2$
the same opinion as at time $t$, or some nodes have changed their opinion to $q$. That is,
all nodes with opinion $q$ at time $t$ will also have opinion $q$ at time
$t+2$. So there are two possibilities. Either every two time steps at least one node switches to
opinion $q$ or every node has again its former opinion and we are in a \sstate. Thus the
process \stabil{s} in at most $2\cdot|\{v \in V: f_t(v) \neq q\}| < 2\cdot |V|$ steps.
\end{proof}

We will now use this result to prove an upper bound on the \dtime that is better
than \autoref{thm:E_upper} for dense graphs.

\begin{proof}[Proof of \autoref{thm:half_E_upper}]
Let $f_t$ be an opinion assignment that has not yet reached a \sstate at time $t$. In the proof of
\autoref{lem:fundamentalist_stabilization} we made the observation that there
must exist a node $v \in V$ for which $f_t(v)\neq f_{t-2}(v)$. We therefore
distinguish the following two cases.

\noindent \textbf{Case 1.} The opinion assignment $f_{t-2}$ is \fund{q}. 

\noindent \textbf{Case 2.} There exists \emph{another} node $u$ with
$f_{t-2}(u)\neq f_{t-2}(v)$ such that $f_t(u)\neq f_{t-2}(u)$, that is, $u$ is
non-\stab and disagrees with $v$ at times $t$ and $t-2$.

\noindent As long as we are in case 2, by repeating the argument in case 2 of
the proof of \autoref{thm:Winkler} we see that the number of bad arrows
drops by at least $2$ in each step (one due to $v$ and another one due to $u$). 
According to \autoref{lem:max_bad_arrows}, this can be the case for at most
$1+\frac{\left(|E|-|V_{odd}|/2\right)}{2}$ steps, since the \rd will \stabil after that time. On the other hand, if at some point we 
are in case 1, the process will \stabil in at most $2\cdot|V|$
steps as shown in \autoref{lem:fundamentalist_stabilization}. Together, these
two cases yield the bound
\begin{equation*}
	1+\frac{|E|-|V_{odd}|/2}{2}+2\cdot|V|=
	1+\frac{|E|}{2}-\frac{|V_{odd}|}{4}+\frac{|V|}{4}+\frac{7}{4}\cdot|V|=
	1+\frac{|E|}{2}+\frac{|\Veven|}{4}+\frac{7}{4}\cdot|V| \enspace . \qedhere
\end{equation*}
\end{proof}

\subsection{Omitted Proofs from \autoref{sect:symmetry}}
\label{apx:influence-of-symmetry}

{ \let\iff\Leftrightarrow
\begin{proof}[Proof of \autoref{lem:transitivity}]
Reflexivity and symmetry of $\sim$ hold trivially.
What remains to be shown is transitivity, that is, $\forall u, v, w \in V$ it holds that if $\family{u}=\family{v}$ and $\family{v}=\family{w}$, then also $\family{u}=\family{w}$.
	By definition, we have 
$		N\left( u
	\right)\setminus \left\{ v \right\}=N\left( v \right)\setminus\left\{ u
	\right\}$ and $
N\left( v \right)\setminus \left\{ w \right\}=N\left( w
	\right)\setminus\left\{ v \right\}
$.
By using the previous identities,
	it follows that $N\left( u \right)\setminus \left\{ w,v
		\right\} =( N\left( u \right)\setminus \left\{ v \right\}
		)\setminus \left\{ w \right\}=( N\left( v \right)\setminus
		\left\{ u \right\})\setminus \left\{ w \right\}= ( N\left( v
		\right)\setminus \left\{ w \right\})\setminus \left\{ u
		\right\}= ( N\left( w
		\right)\setminus \left\{ v
		\right\})\setminus \left\{ u \right\}=N\left( w
		\right)\setminus \left\{ u,v \right\}$
	and 
	\begin{equation} v\in N\left( u \right)\iff u\in
		N\left( v \right) \iff
		u\in N\left( w \right)\iff w\in N\left( u
		\right) \iff w\in
		N\left( v \right)\iff v\in N\left( w \right) \enspace .
		\label{eq:transitivity2} 
	\end{equation} Due to
	\eqref{eq:transitivity2}, $w$ either belongs to both $N\left( u
	\right)$ and $N\left( w \right)$ or to none of them, hence
	\eqref{eq:transitivity2} implies $N\left( u \right)\setminus \left\{ w
	\right\}=N\left( w \right)\setminus\left\{ u \right\}$.
This shows transitivity of the relation $\sim$. 

From the transitive property of $\sim$ it follows that all nodes in the same family either form a clique or are all non-adjacent. From this latter fact together with the definition 
$$\family{u}=\family{v} \iff N\left( u \right)\setminus\left\{ v \right\}=N\left( v \right)\setminus\left\{ u \right\}$$ it also follows that all nodes in the same family have the same degree.
\end{proof}
}

\section{Further Computational Properties}
\label{sec:furthercompprops}
In this appendix we  investigate some properties of the \rd
w.r.t.\ the potential function of \cite{GO80, PS83}, that is, the number of \emph{bad arrows} defined in \autoref{def:bad-arrow}. 
We show that the \contime is not monotone w.r.t.\ the value of the potential function, and we investigate how many
opinion assignments exhibit the same bad arrows. Overall, our results highlight the strengths and weaknesses of such a potential function
approach in bounding the \dtime of the \rd.
\label{sect:computational-properties}
\begin{lemma} \label{lem:not-monotone}
The \dtime is not monotone w.r.t.\ the initial number of bad arrows.
\end{lemma}
\begin{proof}
Let $G$ be a graph consisting of a star graph $S_i$ with $i$ leaves that has a
path graph $P_j$ of length $j$ connected to its center node such that $i >
j$. We now can define two initial opinion assignments $f^{(bad)}$ and
$f^{(good)}$ for which the initial number of bad arrows in $f^{(bad)}$ is
greater than the initial number of bad arrows in $f^{(good)}$ but still
$\ctime(G,f^{(bad)}) < \ctime(G,f^{(good)})$.

As $f^{(bad)}$ assignment, we color $\left\lceil \ifrac{i}{2} \right\rceil - 1$
leaves of the star graph $S_i$ white and all other nodes, including the path
$P_j$, black. As $f^{(good)}$ assignment, we color all the nodes of $S_i$
black and assign alternating opinions to the nodes of the path $P_j$. It is
straightforward to verify that the described opinion assignments prove the
statement.
\end{proof}
\begin{figure}[t!]
\centering
\begin{subfigure}{0.45\textwidth} \centering
\includegraphics{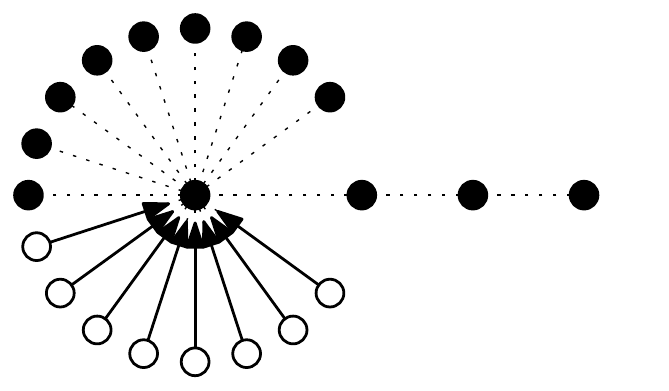}
\caption{initial opinion assignment $f^{(bad)}$}
\label{fig:example-not-monotone-1}
\end{subfigure}
\begin{subfigure}{0.45\textwidth} \centering
\includegraphics{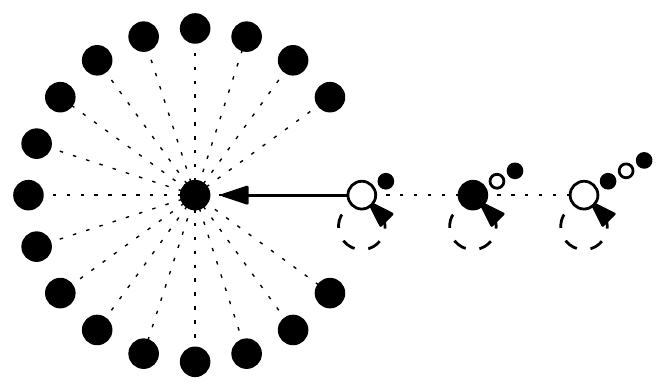}
\caption{initial opinion assignment $f^{(good)}$}
\label{fig:example-not-monotone-2}
\end{subfigure}
\caption{The figure shows an example for the graph described in the proof of
\autoref{lem:not-monotone}. It consists of a star graph $S_{17}$ joined at the
center node with a path of length $3$. Clearly, the initial opinion assignment
$f^{(bad)}$ shown in \autoref{fig:example-not-monotone-1} has a total number of
$7$ bad arrows while the initial opinion assignment $f^{(good)}$ shown in
\autoref{fig:example-not-monotone-1} has only one \emph{true} bad arrow along
with $3$ self loop bad arrows. Still, the process will \stabil{} in only one
step for $f^{(bad)}$ while it will take $3$ steps for $f^{(good)}$.}
\label{fig:example-not-monotone}
\end{figure}

An example for a graph $G$ consisting of a $S_{17}$ and a $P_{3}$
can be seen in \autoref{fig:example-not-monotone}. 
The example shows that even though the initial opinion assignment in \autoref{fig:example-not-monotone-1} has much more initial bad arrows,
the \rd \stabil{s} much faster for the opinion assignment shown in \autoref{fig:example-not-monotone-2}.

Suppose that, instead of specifying the initial opinion assignment, we decide
in advance what bad arrows are there. We can do that by deciding for each
ordered pair $\left( u,v \right)$ for which $\left\{ u,v \right\}\in E$
whether we want to have a bad arrow going from $u$ to $v$. We formalize this
notion by means of the following definitions.

\begin{definition}
Let $G=(V,E)$ be a graph and $\beta : V\times V \rightarrow \left\{0,1\right\}$
denote a characteristic function on $V \times V$. Then $\beta$ is a bad arrows
assignment on $G$ if there exists an opinion assignment $f$ on $G$ that
determines $\beta$ such that $\beta$ is the indicator function of the bad arrows we
have on $G$ w.r.t.\ the opinion assignment $f$.
\end{definition}

According to this definition we clearly have $\left\{u,v\right\} \notin E
\implies \beta(u,v) = 0$ for any bad arrows assignment $\beta$. However, there do also
exist characteristic functions on the (directed) set of edges of $G$ that do
not form a valid bad arrows assignment. An example of such an invalid
assignment that motivates above definition is shown in
\autoref{fig:invalid-bad-arrows}.

\begin{figure}[h!]
\centering
\begin{subfigure}{0.4\textwidth}
\centering
\includegraphics{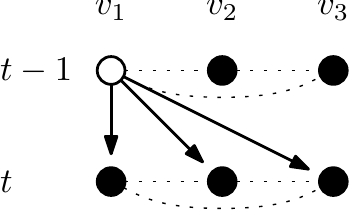}
\caption{valid assignment}
\end{subfigure}
\begin{subfigure}{0.4\textwidth}
\centering
\includegraphics{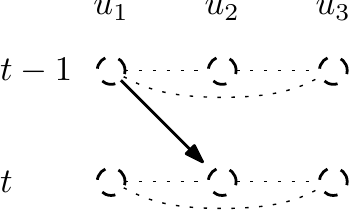}
\caption{invalid assignment}
\end{subfigure}
\caption{Not every characteristic function on the directed edges is a
valid bad arrows assignment.}
\label{fig:invalid-bad-arrows}
\end{figure}

\autoref{fig:invalid-bad-arrows} shows two different assignments
of bad arrows for the $K_3$, a clique of size $3$. 
The left assignment is valid, whereas the right assignment cannot
be valid. This is since in cliques of odd size all nodes share the same opinion
after exactly one step. Therefore all nodes at step $t$ will have the same
opinion. Since, however, $u_1$ had in step $t-1$ a different opinion than this
majority opinion in step $t$, a bad arrow must exist between $u_1$ and
$u_3$ (and also a loop from $u_1$ to itself, if we consider self-loops).

In proving upper bounds on the \dtime we consider the bad arrows assignment determined by the initial opinion
assignment. One may wonder whether in doing so we are \emph{losing
information}. In the following lemma we show that, given a valid bad arrows
assignment, we can reconstruct the initial opinion assignment up to exchanging
black and white (and up to two more possibilities in bipartite graphs).

\begin{lemma} \label{lem:arrows_vs_opinions}
Let $G$ be a connected graph and let $\beta$ be a valid bad arrows assignment on
$G$. If the graph is not bipartite, there are exactly two opinion assignments, otherwise there are exactly
four opinion assignments that determine $\beta$. 
\end{lemma}

\noindent \textit{Proof. }
Let $v \in V$ denote an arbitrary but fixed vertex. We now denote the set
$\{v\}$ as $N_0$ and the set of direct neighbors of $v$ as $N_1$ to define the
$i$-th neighborhood $N_i$ for $i \geq 2$ as 
\begin{equation*}
	N_{i} = \left( \bigcup_{u \in N_{i-1}}\!\!\!N(u) \right) \setminus 
		\left(\bigcup_{j=1}^{i-1}N_j\right) \enspace .
\end{equation*}
That is, the set $N_i$ contains all nodes with shortest path to $v$ of length~$i$.

We now show by an induction on $k = 0,1,2,\dots$ that the colors of all nodes in
$N_{2\cdot k}$ are determined by the color of $v$. The base-case is trivial
since for $k=0$ we have $N_0 = \{v\}$. For the induction step we observe that according to the
induction hypothesis the color of each node in $N_{2\cdot k}$ is determined.
We now observe that the color at time $1$ of each node in $N_{2\cdot k+1}$ is 
determined by $\beta$ and the colors at time $0$ of the nodes in $N_{2\cdot k}$. 
Vice versa, also the colors at time $0$ of nodes in $N_{2\cdot (k+1)}$ 
are determined by $\beta$ and the colors at time $1$ of each node in $N_{2\cdot k+1}$. 
This concludes the induction.

\begin{wrapfigure}{r}{0.33\textwidth}
\centering
\includegraphics{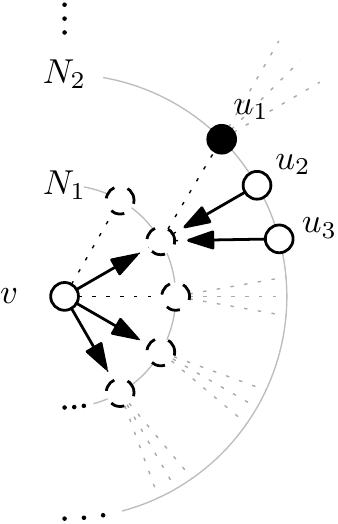}
{
\captionsetup{format=plain,font=small}
\caption{A bad arrows assignment where the opinions of the nodes of each second
neighborhood of the graph are uniquely determined.}
\label{fig:bad-arrows-to-opinion-assignment}
}
\vspace{-1\baselineskip}
\end{wrapfigure}

An example is shown in \autoref{fig:bad-arrows-to-opinion-assignment}. In this
example it is clear that $v$ and, for example, $u_1$ must have a different
color. Since $u_1$ does not have a bad arrow to its neighbor in $N_1$, it has
the same color in the next round as this neighbor. But this neighbor's color in
the next round is different to the current color of $v$ because of the bad
arrow assignment.

Observe that from above induction the lemma follows immediately for bipartite
graphs. We can fix the colors for two arbitrary nodes, one from each of the
two sets of non-adjacent nodes, to determine all other nodes' colors. 
This gives us four possible opinion
assignments for a given bad arrow assignment $\beta$. On the other hand, if the
graph is not bipartite there must exist a cycle of odd length. The opinion
assignments for all nodes of this cycle are determined by $\beta$ with the same
argument as in above induction. Therefore, not only the colors of even
neighborhoods $N_{2\cdot k}$ are determined, but also of odd neighborhoods
$N_{2\cdot k + 1}$. This leaves us with exactly two possible initial opinion
assignments, which concludes the proof. \qed

The following lemma shows that the \dtime does not depend, at least straightforwardly, on the diameter.

\begin{lemma}
For any given graph $G$ with diameter $\Delta$, there exists a graph $G'$ with the following properties:
\begin{itemize}
\item For any opinion assignment $f$ for $G$, there exists and assignment $f'$ for $G'$ such that the \contime of $G$ is the same as in $G'$
\item The diameter of $G'$ is constant
\item $G$ is a subgraph of $G'$.
\end{itemize}
\end{lemma}
\begin{proof}
We augment $G$ by adding a clique $C_0$ of size $n$ where all nodes have
Opinion 0 to $G$. We then add node $u_0$ initialized with $0$ and connect it to
all nodes of $G$ and $C_0$. Symmetrically, we add a clique $C_1$ of size $n$
where all nodes have Opinion 1 to $G$. We then add a node $u_1$ initialized
with $1$ and connect it to all nodes of $G$ and $C_1$. Note that every node
$u\in G$ is also in $G'$ and the opinion of $u$ is the same in both graphs for
any point in time. Hence the \contime remains the same in $G'$ and the claim
follows by observing that $G'$ has a constant diameter.
\end{proof}

\begin{figure}
\centering
\includegraphics{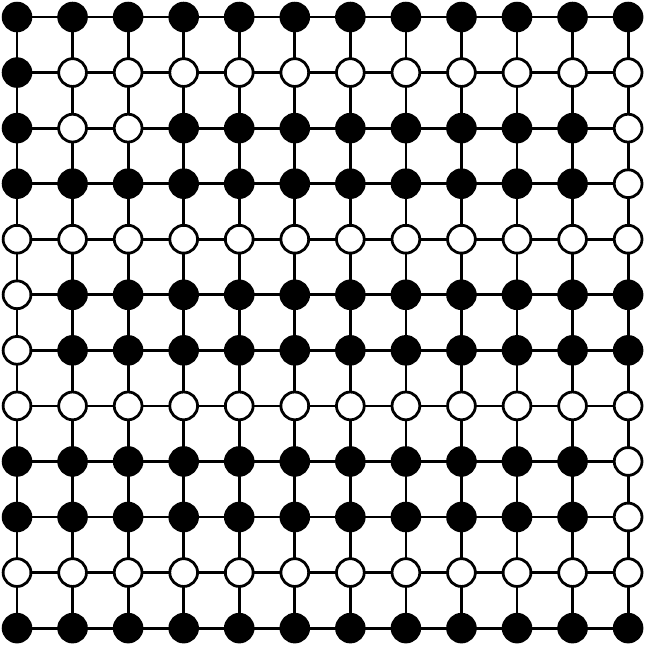}
\caption{A two-dimensional grid $G$ with $\operatorname{diam}(G) = \sqrt{n}$
and an initial opinion assignment that \stabil{s} only after $\BigOmega{n}$
steps.}
\label{fig:example-grid}
\end{figure}

Note that above lemma shows that for any connected graph $G=(V,E)$ and any
initial opinion assignment $f_0$ one can construct another graph $G'$ which has
$G$ as an induced subgraph, asymptotically the same number of nodes and edges,
the same convergence time for a related initial opinion assignment, but a
constant diameter. However, there are even examples of graphs where the
\contime of the \rd w.r.t.\ a given initial opinion assignment $f_0$ is
asymptotically larger than the diameter of the network without modifying the
graph, that is, $\ctime(G, f_0) \in \omega(\operatorname{diam}(G))$.

An example for such a graph is shown in \autoref{fig:example-grid}. In this
example, we are given a two-dimensional grid $G$ of size $|V| = \sqrt{n} \times
\sqrt{n}$. Clearly, the diameter of this graph is $2\cdot\sqrt{n}$. However, by
laying a winding \emph{serpentine} path of white nodes in an entirely black
grid as initial opinion assignment $f_0$ we can force the process to require a
\contime of $\ctime(G, f_0) = \BigOmega{n} \gg \operatorname{diam}(G) =
\BigO{\sqrt{n}}$.

\bigskip
\begin{wrapfigure}{r}{0.5\textwidth}
\vspace{-2.04\baselineskip}
\centering
\begin{subfigure}{0.22\textwidth} \centering
\includegraphics{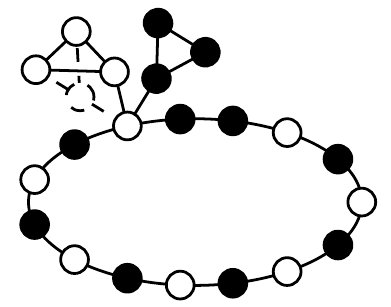}
\caption{original graph $G$}
\label{fig:example-circle-4}
\end{subfigure}
\begin{subfigure}{0.22\textwidth} \centering
\includegraphics{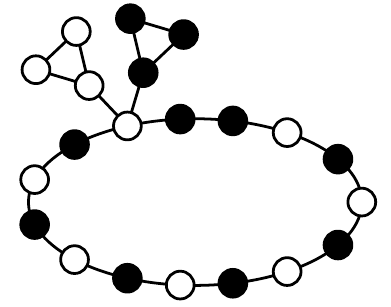}
\caption{graph $G^\Delta$}
\label{fig:example-circle-3}
\end{subfigure}

{
\captionsetup{format=plain,font=small}
\caption{The left graph $G$ is a circle graph with additional gadgets connected
to one node. The dashed node is removed to obtain the
graph $G^\Delta$ shown on the right.}
\label{fig:example-circleX}
}
\vspace{-2\baselineskip}
\end{wrapfigure}

In \autoref{thm:asym_upper}, we show that for the \dtime we have $\max_f\ctime(G^\Delta,f) \geq \max_f\ctime(G, f)$.
However, in general it is not the case that $ \ctime(G^\Delta,f) \geq \ctime(G, f)$
for every opinion assignment $f$, as we show in the following lemma.

\begin{lemma}
Let $G=(V,E)$ be a graph with initial opinion assignment $f$ and $G^\Delta$ be
the asymmetric graph constructed from $G$. In general, it does not hold that
$\ctime(G^\Delta,f) \geq \ctime(G, f)$.
\end{lemma}
\begin{proof}
An example for a graph for which $\ctime(G^\Delta,f) \leq \ctime(G, f)$ is shown
in \autoref{fig:example-circleX}.
\end{proof}

\end{document}